\documentclass[11pt]{article}
\usepackage{fullpage}
\usepackage{algorithmic}
\usepackage{epsfig}
\usepackage{epstopdf}
\usepackage{graphicx}
\usepackage{latexsym}
\usepackage{amsmath}
\usepackage{amsfonts}
\usepackage{amssymb}
\usepackage{mathrsfs}
\usepackage{pifont}
\usepackage{yhmath}
\usepackage{undertilde}
\usepackage{bbm}
\usepackage{eufrak}
\usepackage{amsthm}
\usepackage{booktabs}
\usepackage[usenames,dvipsnames]{color}
\usepackage{stmaryrd}
\usepackage{wasysym}
\usepackage{bbm}
\usepackage{array}
\usepackage[ruled]{algorithm2e}

\usepackage{bigstrut,multirow,rotating}

\usepackage{tabularx, environ}

\makeatletter

\newcolumntype{\expand}{}
\long\@namedef{NC@rewrite@\string\expand}{\expandafter\NC@find}

\NewEnviron{problem}[2][]{%
  \def\problem@arg{#1}%
  \def\problem@framed{framed}%
  \def\problem@lined{lined}%
  \def\problem@doublelined{doublelined}%
  \ifx\problem@arg\@empty%
    \def\problem@hline{}%
  \else%
    \ifx\problem@arg\problem@doublelined%
      \def\problem@hline{\hline\hline}%
    \else%
      \def\problem@hline{\hline}%
    \fi%
  \fi%
  \ifx\problem@arg\problem@framed%
    \def\problem@tablelayout{|>{\bfseries}lX|c}%
    \def\problem@title{\multicolumn{2}{|l|}{%
        \raisebox{-\fboxsep}{\textsc{\Large #2}}%
      }}%
  \else
    \def\problem@tablelayout{>{\bfseries}lXc}%
    \def\problem@title{\multicolumn{2}{l}{%
        \raisebox{-\fboxsep}{\textsc{\Large #2}}%
      }}%
  \fi%
  \bigskip\par\noindent%
  \begin{tabularx}{\textwidth}{\expand\problem@tablelayout}%
    \problem@hline%
    \problem@title\\[2\fboxsep]%
    \BODY\\\problem@hline%
  \end{tabularx}%
  \medskip\par%
}
\makeatother

\usepackage{bigstrut,multirow,rotating}

\newtheorem{theorem}{Theorem}[section]

\newtheorem{corollary}[theorem]{Corollary}
\newtheorem{proposition}[theorem]{Proposition}

\theoremstyle{remark}

\makeatletter
\newcommand\figcaption{\def\@captype{figure}\caption}
\newcommand\tabcaption{\def\@captype{table}\caption}
\makeatother

\newcommand{\bluecomment}[1]{\textcolor{blue}{\textrm{#1}}}

\DeclareMathAlphabet{\mathpzc}{OT1}{pzc}{m}{it}

\begin{document}
\newcounter{my}
\newenvironment{mylabel}
{
\begin{list}{(\roman{my})}{
\setlength{\parsep}{-1mm}
\setlength{\labelwidth}{8mm}
\usecounter{my}}
}{\end{list}}

\newcounter{my2}
\newenvironment{mylabel2}
{
\begin{list}{(\alph{my2})}{
\setlength{\parsep}{-0mm} \setlength{\labelwidth}{8mm}
\setlength{\leftmargin}{3mm}
\usecounter{my2}}
}{\end{list}}

\newcounter{my3}
\newenvironment{mylabel3}
{
\begin{list}{(\alph{my3})}{
\setlength{\parsep}{-1mm}
\setlength{\labelwidth}{8mm}
\setlength{\leftmargin}{10mm}
\usecounter{my3}}
}{\end{list}}

\title{\bf Price of Fairness in Budget Division for\\ Egalitarian Social Welfare}

\date{}
\maketitle
\vspace{-3em}
\begin{center}

\author{Zhongzheng Tang$^{1}$\quad Chenhao Wang$^{2}$\quad Mengqi Zhang$^{3,4}$\\
${}$\\
$^1$ School of Sciences, Beijing University of Posts and\\ Telecommunications, Beijing
100876, China\\
$^2$ Department of Computer Science and Engineering, University of\\ Nebraska-Lincoln, NE, United States\\
$^3$ Academy of Mathematics and Systems Science, Chinese Academy\\ of Sciences, Beijing 100190, China\\
$^4$ School of Mathematical Sciences, University of Chinese Academy\\ of Sciences, Beijing 100049, China\\
\medskip
\{tangzhongzheng,wangch,mqzhang\}@amss.ac.cn}
\end{center}
\vspace{1em}


\begin{abstract}
We study a participatory budgeting problem of aggregating the preferences of agents and dividing a budget over the projects. A budget division solution is a probability distribution over the projects. The main purpose of our study concerns the comparison between the system optimum solution and a fair solution. We are interested in assessing the quality of fair solutions, i.e., in measuring the system efficiency loss under a fair allocation compared to the one that maximizes (egalitarian) social welfare. This indicator is called the price of fairness. We are also interested in the performance of several aggregation rules. Asymptotically tight bounds are provided both for the price of fairness and the efficiency guarantee of aggregation rules.\\

\end{abstract}

\section{Introduction}
Suppose there is a list of possible projects that require funding, and some self-interested agents (citizens, parties or players) have their preferences over the projects. Participatory budgeting is a process of aggregating the preferences of agents, and allocating a fixed budget over projects \cite{cabannes2004participatory,goel2019knapsack,brandt2017rolling}. It allows citizens to identify, discuss, and prioritize public spending projects, and gives them the power to make real decisions about how to allocate part of a municipal or public budget, and how money is spent.  These problems consist in sharing resources so that the agents have high satisfaction, and at the same time the budget should be utilized in an efficient way from a central point of view.

We consider participatory budgeting as a probabilistic voting process \cite{gibbard1977manipulation}, which takes as input agents' preferences and returns a probability distribution over projects. That is, a budget division outcome is a vector of non-negative numbers, one for each project, summing up to 1. We focus on an important special case of \emph{dichotomous preferences}: each agent either likes or dislikes each project, and her utility is equal to the fraction of the budget spent on projects she likes. 
Dichotomous preferences are of practical interest because they are easy to elicit. This process is also referred to as \emph{approval voting}, as each voter (agent) specifies a subset of the projects that she ``approves".

The decision-maker is confronted with a system objective of social welfare maximization, and looks for a budget division solution that performs well under the objective. A system optimum is any solution maximizing the social welfare. The utilitarian social welfare is defined as the sum of utilities over all agents, and the egalitarian social welfare is the minimum among the utilities of agents.  On the other hand, it is desirable for a budget division solution to achieve the fairness among agents. Fairness usually concerns comparing the utility gained by one agent to the others' utilities. The concept of fairness is not uniquely defined since it strongly depends on the specific problem setting and also on the agents perception of what a fair solution is. For example, the \emph{Individual Fair Share} requires that each one of the $n$ agents receives a $1/n$ 
share of decision power, so she can ensure an outcome she likes at least $1/n$ 
of the time (or with probability at least $1/n$). 

The system optimum may 
be highly unbalanced, and thus violate the fairness axioms. For instance, it could assign all budget to a single project and this may have severe negative effects in many application scenarios.
Thus, it would be beneficial to reach a certain degree of agents' satisfaction by implementing some criterion of fairness. Clearly, the maximum utility of fair solutions in general deviates from the system optimum and thus incurs a loss of system efficiency. In this paper, we want to analyze such a loss implied by a fair solution from a worst-case point of view. 


We are interested in assessing the quality of fair solutions, i.e., in measuring the system efficiency loss under a fair allocation compared to the one that maximizes the system objective.
This indicator is called the \emph{price of fairness}. Michorzewski \emph{et al.} \cite{michorzewski2020price} study the price of fairness in participatory budgeting under the 
objective of maximizing the utilitarian social welfare. We consider this problem from an egalitarian perspective.


\smallskip\noindent\textbf{Fairness axioms.}
Given a budget equal to $1$ and $n$ agents, there are some well-studied fairness criteria for a budget division solution (or simply, a distribution) \cite{bogomolnaia2005collective,duddy2015fair,fain2016core15,aziz2019fair}. Individual Fair Share (IFS) requires that the utility of each agent is at least $1/n$. Stronger fairness properties require that groups are fair in a sense. \emph{Unanimous Fair Share (UFS)}  gives to every group of agents with the same approval set an influence proportional to its size, that is, each agent in this kind of group will obtain a utility at least the group's related size (its size divided by $n$). \emph{Group Fair Share (GFS)} requires that for any group of agents, the total fraction of the projects approved by the agents of this group is at least its relative size. \emph{Core Fair Share (CFS)} reflects the incentive effect in the voting process. It says that for any group, each agent of the group cannot obtain a higher utility under another mixture with a probability proportional to the group size. 
\emph{Average Fair Share (AFS)} requires that the average utility of any group with a common approved outcome is at least its relative size.
A distribution satisfies \emph{implementability (IMP)} if it can be decomposed 
into individual distributions such that no agent is asked to spend budgets on projects she considers as valueless.
We remark that all other axioms mentioned above are stronger than IFS. Besides, CFS, AFS and IMP implies GFS, which implies UFS \cite{aziz2019fair}. 

\smallskip\noindent\textbf{Voting rules.} The input of voting rules, also called participatory budgeting rules, includes a list of possible projects, the total available budget, and the 
preferences of agents. The output is a partition of budget among the projects - determining how much budget to allocate to each project- which can be seen as a distribution. We say a voting rule satisfies a certain fairness axiom, if any distribution induced by this rule satisfies that. We study the following voting rules that have been proposed for this setting.

\emph{Utilitarian (UTIL)} rule selects a distribution maximizing the utilitarian social welfare, which focus only on efficiency. \emph{Conditional Utilitarian (CUT)} rule is its simple variant, maximizing utilitarian welfare subject to the constraint that the returned distribution {is} implementable. 
\emph{Random Priority (RP)} rule  averages the outcomes of all deterministic priority rules. \emph{Nash Max Product (NASH)} rule balances well the efficiency and fairness, which selects the distribution maximizing the product of agents' utilities. 
\emph{Egalitarian (EGAL)} rule selects a distribution maximizing the minimum utility of agents. Though it is fair to individuals, it does not attempt to be fair to groups, and the egalitarian objective even treats different voters with the same approval set as if they were a single voter. \emph{Point Voting (PV)} rule  assigns each project a fraction of budget proportional to its approval score, which does not satisfy any of our fairness properties. \emph{Uncoordinated Equal Shares (ES)} rule allocates each agent a $1/n$ share of the budget to spend equally on her approval projects. 

\subsection{Our results}
In this paper, we study the participatory budgeting problem under the objective of maximizing the egalitarian social welfare, i.e., the minimum utility among all agents. Two questions are considered {in a worst-case analysis framework}: how well can a distribution perform on the system efficiency, subject to a fairness constraint, and  {how much social welfare can be achieved} by a certain voting rule. 
Suppose there are $n$ agents and $m$ projects, {and the total budget is 1.}

For the former question, we measure the system efficiency loss under a fair distribution by the price of fairness, defined as the ratio of the social welfare of the best fair distribution to
the social welfare of the system optimum, under the worst-case instance. 
We study six fairness axioms concerning the price of fairness, and provided asymptotically tight bounds in Section \ref{sec:fair}. Because every system optimum satisfies IFS, the price of IFS is trivially 1. By constructing an example,
we show that no distribution satisfying UFS (or GFS, IMP, AFS, CFS) can do better than $\frac2n$ for this example, and prove (almost) tight lower bounds. Our results are summarized in Table \ref{tab:1}.

\begin{table}[htbp]
  \centering
  \caption{\normalsize The price of fairness for 6 axioms.} \label{tab:1}
    \begin{tabular}{p{3 cm}<{\centering}|p{1.3cm}<{\centering}|p{1.3cm}<{\centering}|p{1.3cm}<{\centering}|p{1.3cm}<{\centering}|p{1.3cm}<{\centering}|p{1.3cm}<{\centering}}
    \hline
     Fairness axioms &  {IFS}  &{UFS} &  {GFS}  &  {IMP}  &  {AFS}  &  {CFS} \bigstrut\\
    \hline
     Lower bounds & {1}&{$\frac2n$}&{$\frac{2}{n}$}& $\frac2n-\frac{1}{n^2}$  & $\frac2n-\frac{1}{n^2}$  & $\frac2n-\frac{1}{n^2}$   \bigstrut[b]\\
    \hline
    Upper bounds & {1}&{$\frac2n$}&{$\frac{2}{n}$}&{$\frac{2}{n}$}&{$\frac{2}{n}$}&{$\frac{2}{n}$}\bigstrut[b]\\
    \hline
    \end{tabular}
\end{table}

For the latter question, we study seven voting rules in Section \ref{sec:rule}. The \emph{efficiency guarantee} \cite{lackner2019quantitative} of a voting rule is the worst-case ratio between the social welfare induced by the rule 
and the system optimum.  We provide asymptotically tight bounds for their efficiency guarantees, as shown in Table \ref{tab:2}. Obviously EGAL is optimal, but it is not fair enough. CUT, NASH, ES and RP have a guarantee of $\Theta(\frac{1}{n})$, and in particular, NASH is very fair in the sense that it satisfies all axioms mentioned above.

\begin{table}[htbp]
\centering
\caption{\normalsize Efficiency guarantees for 7 voting rules}\label{tab:2}
\begin{tabular}{p{3cm}<{\centering}|p{1.2cm}<{\centering}|p{1.2cm}<{\centering}|p{1.2cm}<{\centering}|p{1.2cm}<{\centering}|p{1.2cm}<{\centering}|p{1.6cm}<{\centering}|p{1.1cm}<{\centering}}
\hline
{Voting rules}&{UTIL}&{CUT}&{NASH}&{EGAL}&{PV}&{ES}&{RP}\bigstrut\\
\hline
{Lower bounds}&{0}&{$\frac{1}{n}$}&   $\frac2n-\frac{1}{n^2}$   & 1 &{0}&{$\frac{1}{n}$}&$\frac{2}{n}$ \bigstrut[b]\\
\hline
{Upper bounds}&{0}&{$\frac{1}{n-3}$}&{$\frac{2}{n}$}&{1}& $O(\frac{1}{mn})$ & $\frac1n+O(\frac{1}{n^k})$ &{$\frac{2}{n}$}\bigstrut[b]\\
\hline
\end{tabular}
\end{table}


\subsection{Related work}
Participatory budgeting (PB), introduced by Cabannes \cite{cabannes2004participatory}, is a process of democratic deliberation and decision-making, in which an authority  allocates a fixed budget to projects, according to the preferences of multiple agents over the projects. Goel \emph{et al.} \cite{goel2019knapsack} and Benade \emph{et al.} \cite{benade2017preference} mainly focus on aggregation rules of PB
for  social welfare maximization. In the setting where the budget is perfectly divisible,  it can be regarded as a probabilistic voting process \cite{gibbard1977manipulation,brandt2017rolling}, where a voting rule takes as input agents' (aka voters') preferences and returns a probability distribution over projects.

An important consideration on PB is what input format to use for preference elicitation - how each voter should express her preferences over the projects. While arbitrary preferences can be complicated and difficult to elicit, dichotomous preferences  are simple and practical \cite{bogomolnaia2004random,bogomolnaia2005collective}, where agents either approve or disapprove a project. For the dichotomous preference, there have been works both for divisible projects (e.g., \cite{bogomolnaia2005collective,aziz2019fair}) and indivisible projects (e.g., \cite{aziz2017proportionally,talmon2019framework}).   This divisible approval-based setup is a popular input format in many settings, for which many fairness notions and voting rules have been proposed \cite{duddy2015fair,aziz2019fair,brandl2019donor}.
The fair share guarantee principles (e.g., IFS, GFS and AFS) are central to the fair division literature \cite{taylor2004fair,bouveret2016fair}.  
IMP is discussed in \cite{brandl2019donor}. Brandl \emph{et.al} \cite{brandl2015incentives} give a formal study of strict participation in probabilistic voting. Recently,  Aziz \emph{et.al} \cite{aziz2019fair} give a detailed discussion of the above fairness notions.

For the voting rules (sometimes referred to as PB algorithms), EGAL rule maximizes the egalitarian social welfare, and  is used as the lead rule in related assignment model with dichotomous preferences in \cite{bogomolnaia2004random}. NASH rule maximizes a classic collective utility function, and has featured prominently in researches \cite{aziz2019fair,caragiannis2019unreasonable11,conitzer2017fair12,fain2016core15}. 
CUT rule was first implicitly used in \cite{duddy2015fair} and studied in more detail by Aziz \emph{et al.} \cite{aziz2019fair}. RP rule is discussed in \cite{bogomolnaia2005collective}.

Our work takes direct inspiration from Michorzewski \emph{et al.} \cite{michorzewski2020price}, who study the price of fairness in the divisible approval-based setting for 
maximizing utilitarian social welfare (while we consider the egalitarian one). Price of fairness 
quantifies the trade-off between fairness properties and maximization of egalitarian social welfare, and is widely studied \cite{bertsimas2011price,aumann2015efficiency,suksompong2019fairly}.


\section{Preliminaries}\label{sec:pre}

An instance is a triple $I=(N,P,A)$, where $N=\{1,\ldots,n\}$ is a set of agents and $P=\{p_1,\ldots,p_m\}$ is a set of projects. Each agent $i\in N$ has an \emph{approval set} $A_i\subseteq P$ over the projects, and $A=\{A_1,\ldots,A_n\}$ is the profile of approval sets. Let $\mathcal{I}_n$ be the set of all instances with $n$ agents.
For each project $p_j\in P$, let $N(p_j)=\{i\in N: p_j\in A_i\}$ be the set of agents who approve $p_j$, and $|N(p_j)|$ be the \emph{approval score} of $p_j$.

A budget division solution is a distribution $\mathbf{x}\in [0,1]^m$ over the projects set $P$, where $x_j$ indicates the budget assigned to project $p_j$, and $\sum_{j=1}^{m}x_j=1$. 
Let $\Delta(P)$ be the set of such distributions. The \emph{utility} of agent $i\in N$ under distribution $\mathbf x$ is {the amount of budget assigned to her approved projects,} 
that is, $u_i(\mathbf{x})=\sum_{p_j\in A_i}x_j.$
The \emph{(egalitarian) social welfare} of $\mathbf x$ is
$$sw(I,\mathbf x)=\min_{i\in N}u_i(\mathbf x).$$
Define the \emph{normalized social welfare} of $\mathbf x$ as
$$\hat{sw}(I,\mathbf x)=\frac{sw(I,\mathbf x)}{sw^*(I)},$$
where {$sw^*(I)$ is the optimal social welfare of instance $I$.}
Clearly, $\hat{sw}(I,\mathbf x)\in[0,1]$. Though the system optimum (that maximizes the minimum utility of agents) is fair in some sense, it is not fair enough. We consider six fairness axioms. 
Given an instance $I=(N,P,A)$, a distribution $\mathbf x$ satisfies
\begin{itemize}
  \item Individual Fair Share (IFS) if $u_i(\mathbf x)\ge 1/n$ for all agent $i\in N$;
  \item {Unanimous Fair Share (UFS)} if for every $S\subseteq N$ such that {$A_i=A_j ~\text{for all}~ i,j \in S$,} we have $u_i(\mathbf x)\ge |S|/n$ 
      for any $i\in S$.
  \item Group Fair Share (GFS) if for every $S\subseteq N$, we have $\sum_{p_j\in \cup_{i\in S}A_i}x_j\ge |S|/n$;
  \item Implementability (IMP) if we can write $\mathbf x=\frac{1}{n}\sum_{i\in N}\mathbf {x_i}$ for some distribution $\mathbf{x_i}$ such that $x_{i,j}>0$ only if $p_j\in A_i$;
  \item Average Fair Share (AFS) if for every $S\subseteq N$ such that $\bigcap_{i\in S}A_i \neq \emptyset$, we have $\frac{1}{|S|}\sum_{i\in S}u_i(\mathbf x)\ge |S|/n$;
  \item Core Fair Share (CFS) if for every $S\subseteq N$, there is no vector $\mathbf z\in [0,1]^m$ with $\sum_{j=1}^{m}z_{j}=|S|/n$ such that $u_i(\mathbf z)>u_i(\mathbf x)$ for all $i\in S$.
\end{itemize}

IFS is the weakest one among the above axioms. Besides, each of CFS, AFS and IMP implies GFS, which implies UFS.

A \emph{voting rule} $f$ is a function that maps an instance $I$ to a distribution $f(I)\in\Delta(P)$. We consider the following voting rules:
\begin{itemize}
\item Utilitarian (UTIL) rule  selects $\mathbf{x}$ maximizing $\sum_{i\in N}u_i(\mathbf{x})$.
\item Conditional Utilitarian (CUT) rule selects the distribution $\frac{1}{n}\sum_{i\in N}\mathbf{x_i}$, where $\mathbf{x_i}$ is the uniform distribution over the projects in $A_i$ with the highest approval score.
\item Nash Max Product (NASH) rule  selects $\mathbf{x}$ maximizing $\prod_{i\in N}u_i(\mathbf{x})$.
\item Egalitarian (EGAL) rule selects $\mathbf{x}$ maximizing $\min_{i\in N}u_i(\mathbf{x})$.
\item Point Voting (PV) rule selects $\mathbf{x}$, where {$x_j=\frac{|N(p_j)|}{\sum_{p\in P}|N(p)|}$} for $ p_j \in P$.
\item Uncoordinated Equal Shares (ES) rule selects distribution $\frac{1}{n}\sum_{i\in N}\mathbf{x_i}$, where $\mathbf{x_i}$ is the uniform distribution over $A_i$, for any $i \in N$.
\item {Random Priority (RP) rule} selects $\frac{1}{n!}\sum_{\sigma \in \Theta(N)}f^\sigma(I)$, where $\Theta(N)$ is the set of all strict orderings of $N$, and {$f^\sigma(I)\in\arg\max_{\mathbf{x}\in \Delta(P)}\succ_{lexico}^{\sigma}$ is a distribution maximizing the utilities of agents lexicographically with ordering $\sigma$.}
\end{itemize}

A voting rule $f$ satisfies a fairness axiom if  distribution $f(I)$ satisfies it for all instances $I$. Table \ref{tab:3} shows the fairness axioms satisfied by the above voting rules.

\begin{table}[ht]
\centering
\caption{\normalsize Fairness axioms satisfied by voting rules}\label{tab:3}
\begin{tabular}{p{2 cm}<{\centering}|p{1cm}<{\centering}|p{1cm}<{\centering}|p{1cm}
<{\centering}|p{1cm}<{\centering}|p{1cm}<{\centering}|p{1cm}<{\centering}|p{1cm}<{\centering}}
\hline
{}&{UTIL}&{CUT}&{NASH}&{EGAL}&{PV}&{ES}&{RP}\bigstrut\\
\hline
{IFS}&{}&{+}&{+}&{+}&{}&{+}&{+}\\
\hline
{UFS}&{}&{+}&{+}&{}&{}&{+}&{+}\\
\hline
{GFS}&{}&{+}&{+}&{}&{}&{+}&{+}\\
\hline
{IMP}&{}&{+}&{+}&{}&{}&{+}&{}\\
\hline
{AFS}&{}&{}&{+}&{}&{}&{}&{}\\
\hline
{CFS}&{}&{}&{+}&{}&{}&{}&{}\\
\hline
\end{tabular}
\end{table}

 As a warm-up, we give some properties on the optimal social welfare.

\begin{proposition}\label{prop:mm}
Let $m^*$ be the minimum possible number such that there is an optimal distribution giving positive budget to exactly $m^*$ projects. If $m^*>1$, the optimal social welfare is at most $\frac{m^*-1}{m^*}$. If $m^*=1$, the optimal social welfare is $1$.
\end{proposition}
\begin{proof}
Consider an optimal distribution $\mathbf x$ that gives positive budget to $m^*>1$ projects. For each project $p_j\in P$, there must exist an agent (say $a_j\in N$) who does not approve $p_j$; otherwise, a distribution allocating budget 1 to this project is optimal, and thus $m^*=1$, a contradiction. Further, because $\mathbf x$ distributes budget 1 among the $m^*$ projects, there is a project $p_k$ receiving a budget at least $\frac{1}{m^*}$, and agent $a_k$ has a utility at most $1-\frac{1}{m^*}$, establishing the proof. 
\end{proof}

\begin{proposition}
Let {$m'=\min_{S\subseteq P:\cup_{p\in S}N(p)=N}|S|$} be the minimum possible number of projects that cover all agents. 
  Then the optimal social welfare is at least $\frac{1}{m'}$.
\end{proposition}
\begin{proof}
Consider $m'$ projects that cover all agents, i.e., each agent approves at least one of the $m'$ projects. A distribution
that allocates $\frac{1}{m'}$ to each of the $m'$ projects induces a utility of at least $\frac{1}{m'}$ for every agent, implying the optimal social welfare at least $\frac{1}{m'}$.
\end{proof}

\begin{proposition}\label{prop:kk}
For an instance $(N, P, A)$, if the optimal social welfare is $\frac{k}{n}$ for some $k\le n$, then there exists a project $p_j\in P$ such that at least $\lfloor k\rfloor$ agents approve it, i.e., $N(p_j)\ge \lfloor k\rfloor$.
\end{proposition}
\begin{proof}
Suppose for contradiction that for every $p_j\in P$, $N(p_j)\le \lfloor k\rfloor-1$. Let $\mathbf x$ be {an optimal distribution}. Each project $p_j$ can provide totally at most $(\lfloor k\rfloor-1)x_j$ utility for the $n$ agents. As $\sum_{p_j\in P}x_j=1$, the total utility that the $m$ projects can provide for the $n$ agents is at most $\lfloor k\rfloor-1$. Hence, there exists at least one agent whose utility is at most $\frac{\lfloor k\rfloor-1}{n}<\frac{k}{n}$, a contradiction with the optimal social welfare.
\end{proof}

\section{Guarantees for fairness axioms}\label{sec:fair}
{Given an instance $I$, the \emph{price of fairness (POF)} of IFS with respect to $I$ is defined as the ratio of the social welfare of the best IFS  distribution to the optimal social welfare, that is,
$$\text{POF}_{IFS}(I)=\sup_{\mathbf x\in\Delta_{IFS}}\frac{sw(I,\mathbf x)}{sw^*(I)}=\sup_{\mathbf x\in\Delta_{IFS}}\hat{sw}(I,\mathbf x),$$
where $\Delta_{IFS}$ is the set of distributions satisfying IFS.

The POF of IFS  is the infimum over all instances, that is,
$$\text{POF}_{IFS}=\inf_{I\in\mathcal I_n}\text{POF}_{IFS}(I).$$
The POFs of other fairness axioms are similarly defined.}

By the definition of IFS (that every agent receives a utility at least $1/n$), 
it is easy to see that every instance admits an IFS distribution, and thus an optimal distribution must satisfy IFS. We immediately have the following theorem.

\begin{theorem}
For any instance $I$, there exists an IFS distribution $\mathbf x$ such that $\hat{sw}(I,\mathbf x)= 1$. {That is, $\text{POF}_{IFS}=1$.} 
\end{theorem}

Also, we can give a tight lower bound for the normalized social welfare of IFS distributions. Recall that GFS implies IFS.
\begin{theorem}
  For any instance $I$ and any IFS (or GFS) distribution $\mathbf x$,
  we have $\hat{sw}(I,\mathbf x)\ge \frac 1 n$.
Further, there exists an instance $I$ and a GFS distribution $\mathbf x$ such that $\hat{sw}(I,\mathbf x)= \frac 1 n$.
\end{theorem}
\begin{proof}
The first claim is straightforward from the definition. For the second claim, we consider an instance $I$ with $n$ agents and $m=2n+1$ projects. For any $i\in N\setminus\{n\}$, the approval set of agent $i$ is $A_i=\{p_{2i-1},p_{2i},p_{2i+1},p_{2n+1}\}$, and $A_n=\{p_{2n-1},p_{2n},p_{1},p_{2n+1}\}$.
That is, all agents have a common approval project $p_{2n+1}$, and each agent $i$ has an approval project $p_{2i}$, which is not approved by other agents. The optimal social welfare is $sw^*(I)=1$, attained by placing all budget to common project $p_{2n+1}$. Consider a distribution $\mathbf x$ where $x_{2i}=\frac 1 n$ for each $i\in N$, and $x_j=0$ for {any} 
other project $p_j$. The utility of every agent is $\frac 1 n$, and it is easy to check that $\mathbf x$ satisfies GFS, {because for any group  $S\subseteq N$ we have $\sum_{p_j\in \cup_{i\in S}A_i}x_j= |S|/n$.} Then the social welfare induced by distribution $\mathbf x$ is $sw(I,\mathbf x)=\frac 1 n$, which implies $\hat{sw}(I,x)=\frac{sw(I,x)}{sw^*(I)}=\frac 1 n$. This completes the proof.
\end{proof}

In the following we give a {universal upper bound $\frac 2n$} on the POFs of all other fairness axioms.

\begin{theorem}\label{thm:ins}
  There exists an instance $I$ such that for every distribution $\mathbf x$ satisfying UFS (or GFS, IMP, AFS, CFS), we have $\hat{sw}(I,\mathbf x)\le \frac 2 n$. That is, the POF of  UFS (or GFS, IMP, AFS, CFS)  is at most $\frac2n$.  
\end{theorem}
\begin{proof}
  Consider an instance $I$ with $n$ agents and $2$ projects. Agents $1,2,\ldots,n-1$ approve project $p_1$, and agent $n$ approves $p_2$. That is, $N(p_1)=\{1,2,\ldots,n-1\}$ and $N(p_2)=\{n\}$. The optimal social welfare is $sw^*(I)=1/2$, attained by giving each project half of the budget. For any distribution $\mathbf x$ satisfying UFS for instance $I$, let the utility of each agent {in  $N(p_1)$} be $x_1$, and the utility of agent $n$ be $x_2$. Applying {UFS to coalition $N(p_1)$ and $N(p_2)$, respectively, we have $x_1\ge \frac{n-1}{n}$ and $x_2\ge \frac{1}{n}$.}
  Since $x_1+x_2=1$, it must be $x_1=\frac{n-1}{n}, x_2=\frac{1}{n}$. Then $sw(I,\mathbf x)=\min\{x_1,x_2\}=\frac{1}{n}$, and $\hat{sw}(I,x)=\frac{sw(I,x)}{sw^*(I)}=\frac{2}{n}$. This completes the proof for {UFS}. Since each of GFS, IMP, AFS and CFS implies UFS, this conclusion also holds for GFS, IMP, AFS and CFS.
\end{proof}

\section{Guarantees for voting rules}\label{sec:rule}
{In this section, we consider seven voting rules (UTIL, CUT, NASH, EGAL, PV, ES and RP), and analyze their performance on the system objective in the worst case. Further,
these analysis turn back to provide POF results for fairness axioms, as the voting rules satisfy some certain fairness axioms (see Table \ref{tab:3}).}
Define the \emph{efficiency guarantee} (or simply, \emph{guarantee}) of voting rule $f$ {as the worst-case normalized social welfare:}
$$k_{eff}(f)=\min_{I\in \mathcal{I}_n}\hat{sw}(I,f(I)).$$


\begin{theorem}\label{thm:lbe}
The efficiency guarantee of UTIL is 0, and that of EGAL is 1. The efficiency guarantees of {CUT, NASH, ES, and RP} are {all in $[\frac1n,\frac2n]$.}
\end{theorem}
\begin{proof}
The efficiency guarantee of EGAL is trivial. 
Consider the instance constructed in the proof of Theorem \ref{thm:ins}, {where $N(p_1)=\{1,\ldots,n-1\}$ and $N(p_2)=\{n\}$}. The optimal social welfare is $\frac 12$, attained by allocating $x_1=x_2=\frac12$. Rule UTIL returns $x_1=1$ and $x_2=0$, inducing a utility 0 for agent $n$. So the guarantee of UTIL is 0. Rules CUT, NASH, PV, ES all return $x_1=1-\frac1n$ and $x_2=\frac1n$, inducing a utility $\frac1n$ for agent $n$. So the guarantee of these four rules is at most $\frac{1/n}{1/2}=\frac2n$. {(Indeed, this claim simply follows from Theorem \ref{thm:ins}, since all the four rules satisfy UFS.)}

On the other hand, for any instance $I$ and the distribution $\mathbf{x}$ returned by CUT (resp., NASH, ES, RP), we have $u_i(I)\ge \frac{1}{n}$ for any $i\in N$, since the rule satisfies IFS. Then $sw(I,\mathbf x)\ge \frac{1}{n}$, which implies $\hat{sw}(I,\mathbf x)\ge \frac{1}{n}$. Therefore, the efficiency guarantees of CUT (resp., NASH, ES, RP) is at least $\frac{1}{n}$.
\end{proof}

\begin{theorem}
The efficiency guarantee of ES is no better than {$\frac1n+O(\frac{1}{n^k})$, for all $k\in \mathbb{N}^{+}$.} 
\end{theorem}
\begin{proof}
Consider an instance with {$n$ agents and} $m=n^{k+1}+1$ projects. Each agent approves $n^k+1$ projects. The intersection of every two approval sets is project $p_m$, implying that $p_m$ is approved by all agents, and the approval score of every other project is 1. It is easy to see that (by Proposition \ref{prop:mm}), the optimal social welfare is 1, attained by allocating all budget to $p_m$. {The outcome of ES is $x_m=\frac{1}{n^k+1}$, and $x_j=\frac{1}{n}\cdot \frac{1}{n^k+1}$ for any $p_j\neq p_m$. Thus,} ES gives each agent a utility of
$$\frac{1}{n(n^k+1)}\cdot n^k+\frac{1}{n^k+1}=\frac{n^{k-1}+1}{n^k+1}=\frac1n+O(\frac{1}{n^k}),$$
which completes the proof.
\end{proof}

\begin{theorem}
The efficiency guarantee of PV is  $O(\frac{1}{mn})$. 
\end{theorem}
\begin{proof}
Consider an instance $I$ with {$n$ agents and} $m$ projects.
Each agent in $N\setminus\{n\}$ approves $p_1,\ldots,p_{m-1}$, and agent $n$ approves $p_m$. That is, the first $m-1$ projects are approved by the first $n-1$ agents, and the last project is approved by the remaining  agent. The optimal social welfare is $sw^*(I)=\frac{1}{2}$, attained by a distribution with $x_1=x_m=\frac12$. However, PV allocates each project in {$P\setminus \{p_m\}$} 
a budget of $\frac{n-1}{(m-1)(n-1)+1}$,
and project $p_m$ a budget of $\frac{1}{(m-1)(n-1)+1}$. Then the social welfare induced by PV is $\frac{1}{(m-1)(n-1)+1}$, which implies the guarantee of PV is at most $\frac{2}{(m-1)(n-1)+1}=O(\frac{1}{mn})$.
\end{proof}

\begin{theorem}
The efficiency guarantee of CUT is no better than $\frac{1}{n-3}$. 
\end{theorem}
\begin{proof}
Consider an instance with $m={n-1\choose2}+1$. Each of the first ${n-1\choose2}$ projects corresponds to a unique pair of the first $n-1$ agents who disapprove it, and all other agents approve it; the last project $p_m$ is approved by the first $n-1$ agents. {That is, each agent in $N\setminus\{n\}$ disapproves $n-2$ projects in $P\setminus\{p_m\}$, and approves all other projects; agent $n$ approves all $m-1$ projects in $P\setminus\{p_m\}$, and disapproves project $p_m$. Fig. 1 
shows a 5-agent example. }

\begin{figure}[htpb]
\begin{center}
\includegraphics[scale=0.25]{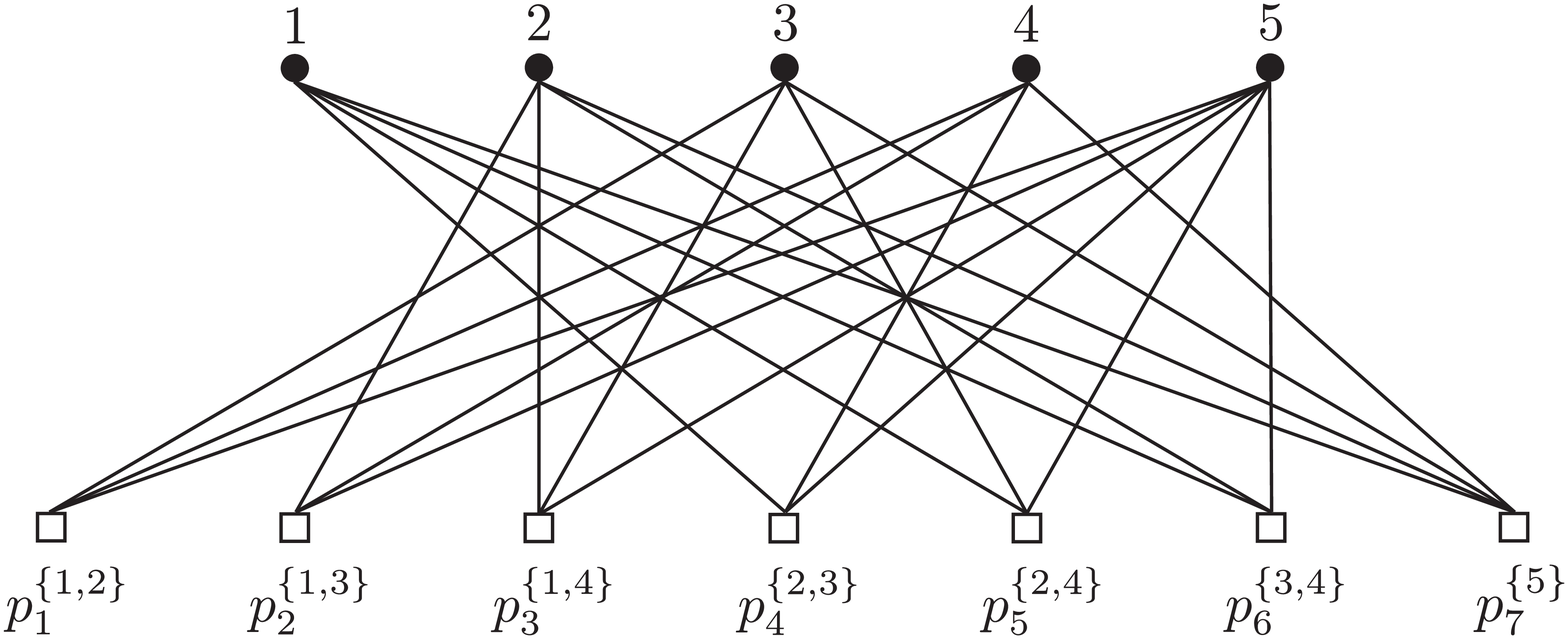}\label{fig:cut}
\caption{{An example with $n=5$ agents and $m=7$ projects, where $p_j^S$ indicates that project $p_j$ is disapproved by every agent in set $S\subseteq N$. Each of the first 6 projects corresponds to a pair of agents. } }
\end{center}
\end{figure}
Then we have $|N(p_m)|=n-1$, and $|N(p_j)|=n-2$ for each $j\leq m-1$. 
The optimal social welfare is at least $1-\frac{n-2}{m} > \frac{n-3}{n-1}$, achieved by allocating uniform budget to each project. However, {CUT rule allocates each project in $P\setminus \{p_m\}$ a budget of $\frac{1}{n(m-1)}$, and project $p_m$ a budget of $\frac{n-1}{n}$. Then the social welfare induced by CUT is $1/n$, (i.e., the utility of agent $n$)}
which implies the efficiency guarantee of CUT is at most $\frac{n-1}{n(n-3)}<\frac{1}{n-3}$.
\end{proof}


\begin{theorem}
The efficiency guarantee of NASH is in $[\frac2n-\frac{1}{n^2},\frac2n]$.
\end{theorem}
\begin{proof}
Let $I$ be an arbitrary instance {with $n$ agents and $m$ projects}, and $f_{NS}(I)$ be {the} 
distribution returned by NASH rule. Since NASH satisfies IFS, we have $sw(I,f_{NS}(I))\ge \frac1n$. By Theorem \ref{thm:ins} and the fact that NASH satisfies UFS, the efficiency guarantee is at most $\frac2n$. 
If the social welfare induced by NASH rule is $sw(I,f_{NS}(I))\ge \frac2n$, the proof is done. So we only need to  consider the case $sw(I,f_{NS}(I))\in[\frac1n, \frac2n)$. Suppose for contradiction that
\begin{equation}\label{eq:nash}
\frac{sw(I,f_{NS}(I))}{sw^*(I)}<\frac2n-\frac{1}{n^2}.
\end{equation}

Let {$\mathbf u^*$ and $\mathbf u^{NS}$} be the utility profiles induced by an optimal distribution and the solution output by NASH, respectively. Let $\bar i\in N$ be the agent with the minimum utility in the NASH solution, i.e., $u_{\bar i}^{NS}=sw(I,f_{NS}(I))$. Then we have
\begin{align}
\sum_{i\in N}\frac{u^*_i}{u_i^{NS}} & = \frac{u^*_{\bar i}}{u_{\bar i}^{NS}} +\sum_{i\in N:i\neq \bar i}\frac{u^*_i}{u_i^{NS}}\nonumber\\
&\ge\frac{sw^*(I)}{sw(I,f_{NS}(I))}+\sum_{i\in N:i\neq \bar i}\frac{u^*_i}{u_i^{NS}}\label{eq:22}\\
&> \frac{1}{2/n-1/n^2}+\frac{sw^*(I)}{1}(n-1)\nonumber\\
&>\frac{1}{2/n-1/n^2}+\frac{1/n}{2/n-1/n^2}(n-1)\label{eq:33}\\
&=n,\label{eq:con}
\end{align}
{where 
(\ref{eq:22}) comes from (\ref{eq:nash}), and 
(\ref{eq:33}) comes from (\ref{eq:nash}) and 
$sw(I,f_{NS}(I))\ge \frac1n$.}

The outcome of NASH rule is an optimal solution of
the following convex optimization problem:
\begin{align*}
\max ~~~&h(\mathbf u) = \sum_{i\in N} \log (u_i)\\
\text{s.t.}~~~&\sum_{p_j\in A_i} x_j = u_i, ~~~ \text{for all~}i \in N\\
&\sum_{j=1}^m x_j = 1,\\
&x_j\geq 0, ~~~\text{for~} j = 1, \dots, m
\end{align*}
Let $\mathcal D$ be the feasible domain of this problem. Since all constraints are linear,  $\mathcal D$ is a convex set. {Let $\mathcal D_{\mathbf u}=\{\mathbf u|\exists~\mathbf x~\text{s.t.}~(\mathbf u, \mathbf x)\in \mathcal D\}$ be a restriction on $\mathbf u$. Then for any $0\leq \alpha\leq  1$ and any utility profile $\mathbf u' \in \mathcal D_{\mathbf u}$, we have $\mathbf u^{NS} + \alpha( \mathbf u' - \mathbf u^{NS})\in\mathcal D_{\mathbf u}$.}
Then, we can derive
\begin{align*}
&\lim_{\alpha\rightarrow0^{+}}\frac{h(\mathbf u^{NS} + \alpha( \mathbf u' - \mathbf u^{NS}))-h(\mathbf u^{NS})}{\alpha}\leq 0\\
\Longrightarrow~ &\nabla h(\mathbf u^{NS})^{\mathrm T}(\mathbf u' - \mathbf u^{NS})\leq 0\\
\Longrightarrow~ &\sum_{i\in N} \frac{u'_i}{u^{NS}_i}\leq n,
\end{align*}
which gives a contradiction to Equation (\ref{eq:con}).
\end{proof}

Because NASH rule satisfies the properties IMP, AFS and CFS, combining with Theorem \ref{thm:ins}, we have the following corollary.

\begin{corollary}
The POFs of IMP, AFS and CFS are all in $[\frac2n-\frac{1}{n^2},\frac2n]$.
\end{corollary}

\begin{theorem}
{The efficiency guarantee of RP is $\frac2n$.}
\end{theorem}
\begin{proof}
Let $I$ be an arbitrary instance, and $f_{RP}(I)$ be {the}
distribution returned by RP rule.
Since RP satisfies IFS, the social welfare of $f_{RP}(I)$ is at least $\frac1n$. If $sw^*(I)\le \frac12$, the normalized social welfare is $\hat{sw}(I,f_{RP}(I))\ge \frac2n$. So it suffices to consider the case $sw^*(I)>\frac12$.

When $sw^*(I)>\frac12$, if there are two agents $i,j\in N$ such that $A_i\cap A_j=\emptyset$, then no distribution can give both agents a utility larger than $\frac12$, a contradiction. So for any two agents, the intersection of their approval sets is non-empty. For each agent $i\in N$, under RP rule, the probability of ranking the first among the $n!$ permutations is $\frac1n$, where {she} 
receives a utility 1. Suppose $i=\sigma(2)$ and $j=\sigma(1)$ for a permutation $\sigma\in\Theta(N)$. Since RP maximizes the utility of agents lexicographically {with respect to $\sigma$}, it must allocate all budget to their intersection $A_i\cap A_j$, and the utilities of agents $i$ and $j$ both are  1. Note that the probability of ranking the second among the $n!$ permutations for agent $i$ is $\frac1n$. The utility of agent $i$ under RP rule is at least
$$\textbf{Pr}\{i=\sigma(1)\}\cdot 1+\textbf{Pr}\{i=\sigma(2)\}\cdot 1=\frac2n,$$
and the normalized social welfare is also at least $\frac2n$.
Combining with the upper bound in Theorem \ref{thm:lbe}, the efficiency guarantee of RP is $\frac2n$.
\end{proof}

{We remark that it is still open whether RP rule can be implemented in polynomial time.} Because RP rule satisfies GFS, combining with Theorem \ref{thm:ins}, we have the following corollary.

\begin{corollary}
The POF of GFS is $\frac2n$.
\end{corollary}

\section{Conclusion}

We quantify the trade-off between the fairness criteria and the maximization of egalitarian social welfare in a participatory budgeting problem with dichotomous preferences. Compared with the work of Michorzewski \emph{et al.} \cite{michorzewski2020price}, which considers this approval-based setting under the utilitarian social welfare, we additionally study a fairness axiom Unanimous Fair Share (UFS) and a voting rule Random Priority (RP). We present (asymptotically) tight bounds on the price of fairness for six fairness axioms and the efficiency guarantees for seven voting rules. In particular, both NASH and RP rules are guaranteed to provide a roughly $\frac2n$ fraction of the optimum egalitarian social welfare. The NASH solution can be computed by solving a convex program, while RP is unknown to be computed efficiently.

Both the work of \cite{michorzewski2020price} and this paper assume that all agents have dichotomous preferences, and an immediate future research direction would be to study the effect of fairness constraint when agents are allowed to have a more general preference. Another avenue for future research is considering the fairness in participatory budgeting from the projects' perspective. For example, the projects (e.g., facility managers and location owners) have their own thoughts, and may have a payoff from the budget division, for which a good solution should balance the system efficiency and the satisfaction of the agents and projects.  
So it would be interesting to study the trade-off between system efficiency and this new class of fairness concepts.

\section*{Acknowledge}
The authors thank Xiaodong Hu and Xujin Chen for their helpful discussions, and
anonymous referees for their valuable feedback. %

\bibliography{reference}

\end{document}